\documentclass{article}

\usepackage{listings}
\usepackage{algorithm}
\usepackage{algorithmic}
\usepackage{times}
\usepackage{mathtools}
\usepackage{moreverb}
\usepackage{pst-node}
\usepackage{graphicx}
\usepackage{multirow}
\usepackage{amssymb}
\usepackage{subfig}
\usepackage{balance}
\usepackage{stmaryrd}
\usepackage[normalem]{ulem}

\lstdefinestyle{MyListStyle}{
language=C, 
basicstyle=\ttfamily\footnotesize, 
tabsize=3, 
captionpos=b, 
breaklines=true, 
breakatwhitespace=false, 
escapeinside={\%*}{*)}, 
frame=single, 
numbers=left,
numbersep=5pt,
numberstyle=\tiny,
}

\newtheorem{definition}{Definition}
\newtheorem{claim}{Claim}
\newtheorem{proof}{Proof}

\newcommand{\ITPonly}{\ensuremath{\mathit{ITP}}}

\usepackage{amsmath}
\usepackage{xspace}

\begin{document}

\title{Irrelevant Predicates}

%




\maketitle





\section{Property-irrelevant predicates}
Loops (and recursions) are major hurdles in scalability of property 
verification tools (verifiers). Although slicing removes
loops which have no bearing on assert
statement in terms of its value and reachability, sliced programs
still have loops challenging scale up of
the verifiers. If we can transform a program by eliminating 
some of 
such loops, it is more likely that a given verifier succeeds 
on transformed program. Of course the transformation will be useful 
only if results on 
transformed program can be used to get results on original program.

Loops existing in a (backward) sliced program, may compute value of 
some variables that impact outcome of assert expression in 
following ways:

\begin{enumerate}
\item Impact on value of assert expression, possibly through a 
     chain of assignments.
\item Impact on value of some predicates, possibly through a 
     chain of assignments, that 
\begin{enumerate}
\item Value impact the assert statement
\item Influence the reachability of assert statement
\end{enumerate}
\end{enumerate}

Obviously, loops of type (1) can not be eliminated, as they directly impact the
value of assert expression.
However, loops of type (2) can be eliminated if property can be checked even 
after abstracting 
the predicates which these loops are impacting.
Since we
want the transformed program to be useful in deciding the outcome on 
original program, ideally we would like the transformed program to be 
{\em property equivalent} to original program.  So we focus on what kind of
predicates can be abstracted so that loops and computations contributing
to value of such predicates can be eliminated.

Let $C$ be a predicate in a given program $P$ in which a property 
is encoded through an assert statement $A$. 
Let $P'$ be an abstract program obtained from $P$ by 
replacing right hand side  of reaching definitions used for $C$,  with 
a non deterministic value, denoted by `*'.

Following observations are obvious:

\begin{enumerate}
\item If predicate $C$ is loop invariant in concrete program $P$ then so it will be in abstract program $P'$ too.
\item If property holds in abstract program $P'$ then it will hold in 
concrete program $P$ also
\item If property gets violated in abstract program $P'$ then it may or may not get violated in
concrete program $P$.
\end{enumerate}

Predicate $C$ is called {\em irrelevant to property} (\ITPonly), if
abstract program $P'$ is {\em property equivalent} to concrete program $P$. 

Actually if one picks up any predicate $C$ from concrete program $P$ and 
generate abstract program $P'$ in 
the manner mentioned above then  case (1) and (2) will always hold.
It is the case (3) which differentiates an arbitrary predicate 
from
an \ITPonly\ predicate. For predicate $C$ to be \ITPonly, in case (3), 
property should get violated in concrete program $P$ also.
\section{Example}
\begin{figure}
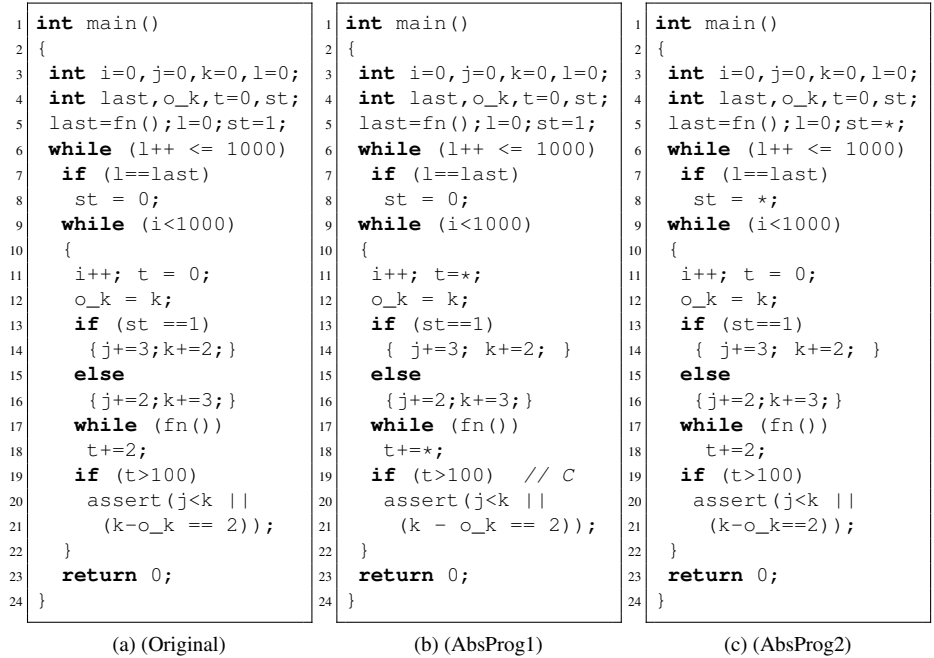

\begin{minipage}{1.40in}
\subfloat[(Original)]{%
\begin{lstlisting}[style=MyListStyle]
int main()
{
 int i=0,j=0,k=0,l=0;
 int last,o_k,t=0,st;
 last=fn();l=0;st=1;
 while (l++ <= 1000)
  if (l==last) 
   st = 0; 
  while (i<1000) 
  {
   i++; t = 0;
   o_k = k;
   if (st ==1)
    {j+=3;k+=2;}
   else 
    {j+=2;k+=3;}
   while (fn()) 
    t+=2;
   if (t>100) 
    assert(j<k || 
     (k-o_k == 2));
  }
  return 0;
}
\end{lstlisting}%
}
\end{minipage}
\hspace{0.15in}
\begin{minipage}{1.40in}
\subfloat[(AbsProg1)]{%
\begin{lstlisting}[style=MyListStyle]
int main()
{
 int i=0,j=0,k=0,l=0;
 int last,o_k,t=0,st;
 last=fn();l=0;st=1;
 while (l++ <= 1000)
  if (l==last) 
   st = 0; 
 while (i<1000) 
 {
  i++; t=*;
  o_k = k;
  if (st==1)
   { j+=3; k+=2; }
  else 
   {j+=2;k+=3;}
  while (fn()) 
   t+=*;
  if (t>100)  // C
   assert(j<k || 
    (k - o_k == 2));
 }
 return 0;
}
\end{lstlisting}%
}
\end{minipage}
\hspace{0.15in}
\begin{minipage}{1.40in}
\subfloat[(AbsProg2)]{%
\begin{lstlisting}[style=MyListStyle]
int main()
{
 int i=0,j=0,k=0,l=0;
 int last,o_k,t=0,st;
 last=fn();l=0;st=*;
 while (l++ <= 1000)
  if (l==last) 
   st = *; 
 while (i<1000) 
 {
  i++; t = 0;
  o_k = k;
  if (st==1)
   { j+=3; k+=2; }
  else 
   {j+=2;k+=3;}
  while (fn()) 
    t+=2;
  if (t>100) 
   assert(j<k || 
    (k-o_k==2));
 }
 return 0;
}
\end{lstlisting}%
}
\end{minipage}
\caption{Abstracting non data impacting predicates}
\label{fig:condabs-example}
\end{figure}

In figure \ref{fig:condabs-example} we show a program annotated as
(Concrete). If we abstract the predicate \texttt{(t > 100)} we will see
that in resulting abstract program, annotated as (AbsProg1), 
property will hold 
and therefore \texttt{(t > 100)} is an \ITPonly\ 
predicate. 

Considering predicate \texttt{(st==1)}, we observe that in concrete program,  
$st$ is
loop invariant for outer while loop.
We abstract 
this predicate also as explained above and resulting abstract program is
shown in Figure \ref{fig:condabs-example} with annotation (AbsPrg2).
It is obvious that the predicate \texttt{(st==1)} is loop
invariant in abstract program too.
The property holds on this
abstract program also and therefore \texttt{(st==1)} is an \ITPonly\ 
predicate. 

Suppose we change the assignment to \texttt{j} at line 16 to \texttt{ j+=3 }in concrete program. Now the
assert will be violated in modified abstract program (AbsProg2), when one assigns a suitable 
value
at non-deterministic assignment to \texttt{st} which makes predicate \texttt{(st==1)} as
$false$. And if an input exists with 
which 
\texttt{st} is assigned value 0 then it will get violated in concrete 
program also. 
It will not get
violated in concrete program only if for no input \texttt{st} gets value 0. Which means
predicate \texttt{st==1} always remains $true$ in concrete program.

\section{A refined definition of \ITPonly\ predicates}
To ensure that loops computing the value of an \ITPonly\ predicate get eliminated, it
will be good to place the non-deterministic assignments at only one
point for a predicate rather than doing so at all reaching definitions 
individually.
We choose this point as the nearest common post dominator of these 
reaching definitions. We claim that such a point exists and
is unique.  Without loss of generality, we assume that this point will be 
on a straight line
segment of CFG(single entry, single exit). 
Let us call such a point as {\em computing point of predicate} $C$ and denote it as 
$\widehat{\mathcal{C}}$. 
In earlier example, we showed that when property holds in concrete program
due to the
predicate $C$ being constant, say $true$, it may not hold in abstract 
program 
because the predicate $C$ can be made $false$ also in abstract 
program. Such predicates can 
never be \ITPonly\ as per the earlier definition. Considering this fact,
we refine our definition of \ITPonly\ predicates using the modified 
abstraction strategy.

Let $C$ be a predicate in a given program $P$ and let $A$ be an 
assertion encoding a property to be checked in program $P$. 
Let $P'$ be an abstract program obtained from (concrete) program $P$ by 
inserting non-deterministic assignments at predicate computing point 
$\widehat{\mathcal{C}}$, 
for variables used in predicate $C$. The scheme is shown diagrammatically 
in Figure \ref{fig:abs-paths}. As can be seen in Figure \ref{fig:abs-paths},
in abstract program $P'$, point $\widehat{C}$ is 
deemed to be after the placement of non-deterministic assignments. 
We claim that abstract program $P'$ is 
a sound abstraction of concrete program $P$.
We make following observations regarding this abstraction mechanism.
\begin{figure}
\centerline{\includegraphics[height=50mm, width=120mm]{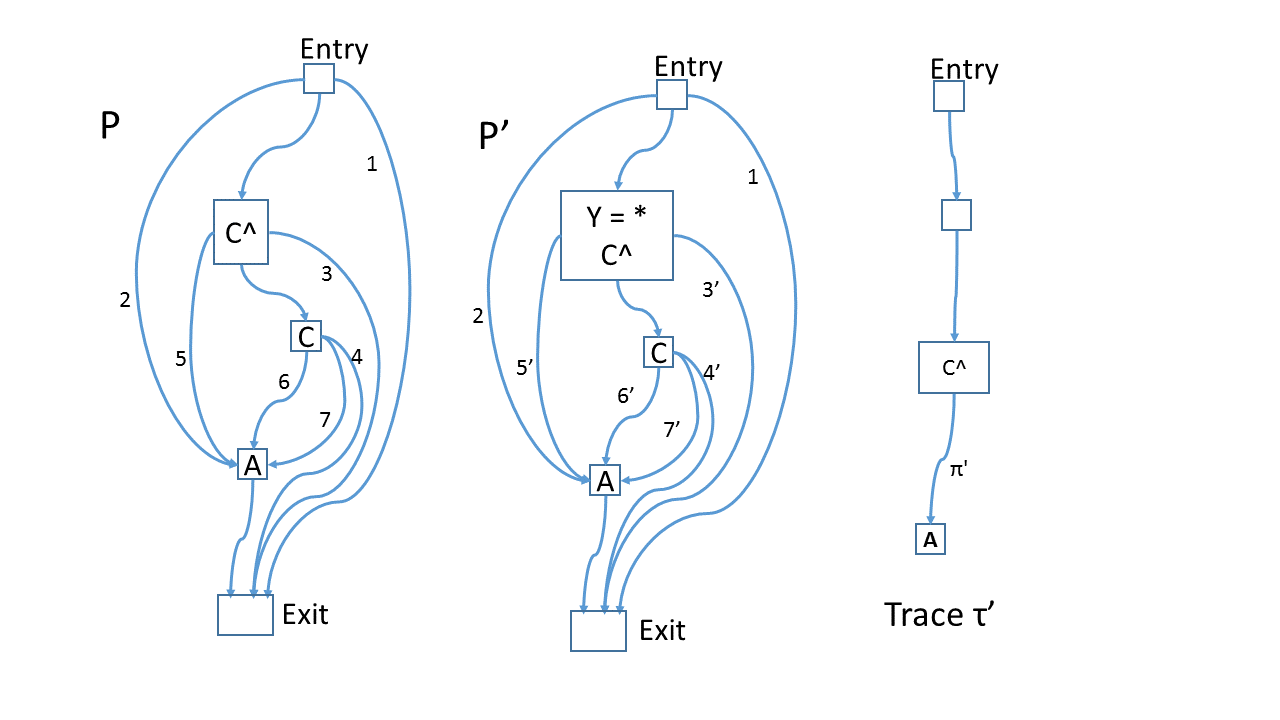}}
\caption{Abstraction and program paths}
\label{fig:abs-paths}
\end{figure}
\begin{enumerate}
\item  Since abstract program $P'$ is a sound abstraction of concrete 
program $P$, if property holds in 
abstract program $P'$, it will hold in concrete program $P$ also.
\item If the property does not hold in abstract program $P'$ then the 
counter example must follow one of the paths labeled as {2,5',6',7'}.
\item Execution traces which bypass 
predicate computing point $\widehat{C}$ will be same in concrete as well 
as abstract program (shown by paths labels 1 and 2 in diagram).
Consequently, if there is a counter example in abstract program $P'$ 
bypassing $\widehat{\mathcal{C}}$, like paths labeled (2) in Figure
\ref{fig:abs-paths}, then the 
same counter example will apply to concrete program $P$ also.
\end{enumerate}

When we have a counter example trace passing throigh $\widehat{C}$, then 
the suffix of the trace starting with last occurrence of $\widehat{C}$ will
be called as {\em violating-suffix}. Obviously all the occurences
of predicate $C$, if any, on the violating-suffix will evaluate to same 
value. A counter example trace will have a violating-suffix if and only
if the trace passes through $\widehat{C}$.

\begin{definition}{(Irrelevant to Property (ITP) Predicates)}
\label{itp_def}
Predicate $C$ is said to be \ITPonly\ for the property encoded with 
assertion $A$ when one of the following holds:
\begin{enumerate}
\item If the property gets violated in abstract program $P'$, with a 
counter example trace having a violating-suffix $\pi$ 
then the property gets violated in concrete program $P$ also with a
counter example trace having a violating-suffix same as $\pi$.
\item Property gets violated in abstract program $P'$ with a counter 
example trace having a violating-suffix with $C$  
evaluating to $b$ within the violating-suffix, and
predicate $C$ never evaluates to $b$ in concrete program $P$.
\end{enumerate}
\end{definition}

\section{A sufficient criterion for \ITPonly\ predicates}
From the Figure \ref{fig:abs-paths}, it is clear that we need to identify 
the predicates for which, the value of assert expression as well as
paths labeled as (5',6' and 7') are not dictated 
by values of variables used in predicate $C$. 
Let $Y$ denote set of such variables. 
In addition, the 
variables which dictate these paths and value of assert  expression, 
should have same values at $\widehat{C}$ in abstract program $P'$
as well as in concrete program $P$.

To formalise the idea we proceed as follows. In the rest of the discussion, 
we assume that program $P$ is already sliced with respect
to assertion $A$. Suppose we abstracted a given
program $P$ to abstract program $P'$ with respect to a predicate $C$,
as per the strategy mentioned earlier. Let 
the property be violated in abstract program $P'$. We want to know under 
what conditions 
we can say that it will be violated in concrete program $P$ also. 
We need to consider only the case where counter example found in 
abstract program passes through the 
predicate computing point $\widehat{C}$.
Let such a counter example trace $\tau'$ 
have program state $\sigma'$ at the last 
occurrence of predicate computing point $\widehat{\mathcal{C}}$ in $\tau'$.
To be precise, $\sigma'$ is the program state just after the sequence of 
non-deterministic assignments placed at $\widehat{C}$.
Let $Y$ be set of variables used in predicate $C$.
Let $\pi'$ be the violating-suffix of trace $\tau'$, as shown in 
Figure \ref{fig:abs-paths}. 

Let $X'$ be set of variables whose value in $\sigma'$ 
determined the  value of $A$ and the
violating-suffix $\pi'$, excluding the control points corresponding to 
predicate $C$.
Intuitively, $X'$ is the set of live variables at 
predicate computing point $\widehat{\mathcal{C}}$, computed as follows:
\begin{enumerate}
\item Start at $A$ with variables used in $A$ as initial set of live
variables.
\item Proceed along path $\pi'$ up to $\widehat{C}$ computing live 
variables at different nodes as per the traditional approach.
\item Treat nodes corresponding to $C$ as identity.
\end{enumerate}

Now, let us see what will it require to get a counter example in concrete 
program $P$ also. There are two cases to be considered.

\begin{enumerate}
\item []{}Case 1: Violating-suffix $\pi'$ bypasses condition $C$.\\
Suppose we get a program state $\sigma$ in concrete program $P$ 
at predicate computing point $\widehat{\mathcal{C}}$ 
such that
it has same values of variables in $X'$ as that in $\sigma'$. Now the path
taken in concrete program $P$, from $\widehat{C}$ starting in state 
$\sigma$ will be same as $\pi'$ and consequently, assertion $A$ will get 
violated in concrete program $P$ also.
\item []{}Case 2: Violating-suffix $\pi'$ passes through condition $C$.\\
Let $b$ be the violating-value of predicate $C$.
Suppose we get a program state $\sigma$ in concrete
program $P$ 
at predicate computing point $\widehat{\mathcal{C}}$ 
such that it has same values 
for variables in $X'$ as that in $\sigma'$ and predicate $C$ evaluates to
$b$ in state $\sigma$. 
It may be noted that value of predicate $C$ will not change 
till we do not revisit the predicate computing point $\widehat{C}$.
Therefore, the path followed from $\widehat{\mathcal{C}}$ starting from
state $\sigma$ will be same as $\pi'$ and consequently, it will 
definitely make  
$A$ get violated in concrete program $P$. 
\end{enumerate}

We generalise above observations along two lines:
\begin{enumerate}
\item We expand $X'$ to $X$ as set of 
live variables along all paths from predicate computing point 
$\widehat{\mathcal{C}}$ to assertion point $A$ 
(with the restriction that $\widehat{\mathcal{C}}$ does not repeat on 
 these paths). 
\item Rather than looking for a state $\sigma$ having same 
$X'$ restriction as that of $\sigma'$, we consider  
restriction of all 
states 
at $\widehat{C}$ 
with respect to $X$( generalisation of $X'$) in abstract program $P'$ 
as well as in concrete program $P$.
\end{enumerate}

To formalise,
let $\Sigma'$  and $\Sigma$ be set of program states at  
$\widehat{\mathcal{C}}$ in abstract program $P'$ and concrete 
program $P$ respectively. 
Let $\Sigma_b \subseteq \Sigma$ be set of states 
in which 
expression of predicate $C$ evaluates to $b \in \{true, false\}$.
We claim as follows:
\begin{claim}
\label{state_cond_for_itp}
If for a predicate $C$, both of the following hold then $C$ is \ITPonly.
\begin{enumerate}
\item [(S1)]{}
$ \lfloor \Sigma' \rfloor_X = \lfloor \Sigma \rfloor_X$
\item [(S2)]{}
$
     \forall b \in \{true, false\} . \Sigma_b \neq \emptyset \implies
     \lfloor \Sigma_b \rfloor_X = \lfloor \Sigma \rfloor_X 
     $
\end{enumerate}
\end{claim}
\begin{proof}
We will consider two cases:
\begin{enumerate}
\item[Case 1:]{}The counter example in $P'$ bypasses the predicate $C$. \\
By the precondition (S1) of the claim, 
 $\lfloor \Sigma' \rfloor_X = \lfloor \Sigma \rfloor_X$. 
 Since $\Sigma' \neq \emptyset$, 
   there will be a state $\sigma \in \Sigma$ 
at $\widehat{C}$ in concrete program $P$ such that $X$ restriction of 
$\sigma'$ will be 
same as $X$ restriction of $\sigma$. 
Therefore path followed from $\sigma$ in concrete program $P$ will be 
exactly same as violating-suffix $\pi'$ and so assertion $A$ will get violated.
\item[Case 2:]{}The counter example in $P'$ passes through $C$.\\
Let $b$ be the violating-value of predicate $C$.
If $C$ always evaluates to constant $\neg b$ in concrete
program $P$ then $C$ is \ITPonly\ as per the
definition ~\ref{itp-def}. 
So we assume that $C$ evaluates to $b$ also sometimes.
Therefore $\Sigma_b \neq \emptyset$. And consequently, by (S2) and (S1),
     $\lfloor \Sigma_b \rfloor_X = \lfloor \Sigma' \rfloor_X$
By similar argument as that in case (1), we can show that 
assertion $A$ will get violated in concrete program $P$ also..
\end{enumerate}
\end{proof}

\section{A computable criterion for \ITPonly\ predicates}
Given an assert $A$ in a program $P$, 
we want to identify predicates in program $P$ which satisfy the two 
conditions of
claim~\ref{state_cond_for_itp}. In these conditions, we talk about values of
$X$ and $Y$ in all program states at predicate computing point
$\widehat{C}$. 


First we consider the condition (S1),
$\lfloor \Sigma' \rfloor_X = \lfloor \Sigma \rfloor_X$.
We observe that starting from same initial state, 
the program state, restricted to $X$, at $\widehat{C}$ at its first 
occurrence
in a trace will be same in abstract program $P'$ and the concrete 
program $P$, provided $X$ and $Y$ are disjoint.
Subsequent changes to program states, restricted to $X$, at $\widehat{C}$ 
will be result of its transformation along the looping paths from this
occurrence of $\widehat{C}$ to its next occurrence. 
We observe that the only difference on paths followed from one occurrence
of $\widehat{C}$ to its next occurrence
in $P$ and $P'$ is new abstract assignments inserted for
$Y$ at $\widehat{C}$ in abstract program $P'$. 
Let $Z$ be the set of live variables at $\widehat{C}$ after traversing
all paths from one occurrence
of $\widehat{C}$ to its next occurrence, after starting with $X$ as set of 
live variables at $\widehat{C}$. It is easy to see that if $Z$ is
disjoint from $Y$ then the transformation
of $X$ from one occurrence of $\widehat{C}$ to its next occurrence will be
in same manner in concrete program as well as abstract program.



\begin{claim}
\label{disjoint_input_vars}
If following two conditions are satisfied then value of $X$ at 
$\widehat{C}$ will be same in $P$ and $P'$.
\begin{enumerate}
\item[(C1)]{} 
    $X$ and $Y$ are disjoint.
\item[(C2)]{}
    The set of live variables on paths from $\widehat{C}$ to 
    $\widehat{C}$ with respect to $\langle X, \widehat{C}\rangle$ is 
    disjoint from $Y$.
\end{enumerate}
\end{claim}

Now we consider the criterion (S2) 
$\lfloor \Sigma_b \rfloor_X = \lfloor \Sigma \rfloor_X$.
Intuitively, this condition 
requires that every $X$ restricted state
of $\Sigma$ is some $X$ restricted state of $\Sigma_b$ also.
Since set of states $\Sigma_b$ will be decided by values of $Y$,
      intuitively, this criteria can be met when $X$ and $Y$ get their
      values in independent manner. That is they are not related in any
      manner when computation proceeds from ENTRY to $\widehat{C}$.

Consider programs given in Figure \ref{fig:related-values}. In program
(a), values of \texttt{x} and \texttt{y} at line 5 get related (they will 
be same as \texttt{z}).
In program (b) values of \texttt{x} and \texttt{y} get related as they 
change together
in the enclosing loop. In program (c) values of \texttt{x} and \texttt{y}
are not related.

\begin{figure}
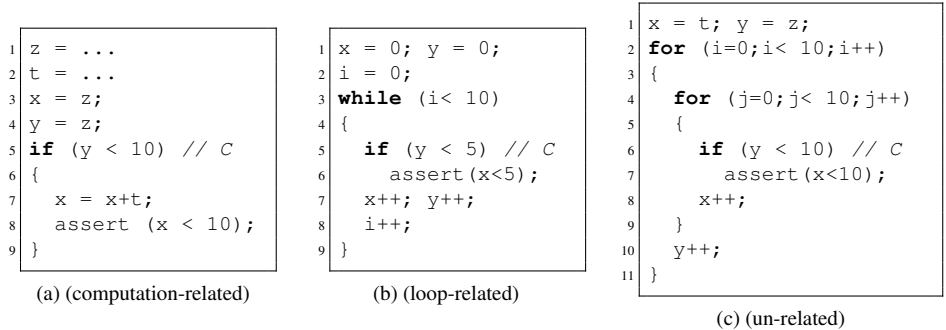

\begin{minipage}{1.25in}
\subfloat[(computation-related)]{%
\begin{lstlisting}[style=MyListStyle]
z = ...
t = ...
x = z;
y = z;
if (y < 10) // C
{
  x = x+t;
  assert (x < 10);
}
\end{lstlisting}%
}
\end{minipage}
\hspace{0.30in}
\begin{minipage}{1.25in}
\subfloat[(loop-related)]{%
\begin{lstlisting}[style=MyListStyle]
x = 0; y = 0;
i = 0;
while (i< 10)
{
  if (y < 5) // C 
    assert(x<5);
  x++; y++;
  i++;
}
\end{lstlisting}%
}
\end{minipage}
\hspace{0.30in}
\begin{minipage}{1.50in}
\subfloat[(un-related)]{%
\begin{lstlisting}[style=MyListStyle]
x = t; y = z;
for (i=0;i< 10;i++)
{
  for (j=0;j< 10;j++)
  {
    if (y < 10) // C
      assert(x<10);
    x++;
  }
  y++;
}
\end{lstlisting}%
}
\end{minipage}
\caption{Illustration of values getting related}
\label{fig:related-values}
\end{figure}

Intuitively, we can achieve independence of $X$ and $Y$ by ensuring 
following:
\begin{enumerate}
\item No computation impacts value of both $X$ as well as $Y$.
\item Values of $(X, \widehat{C})$ and $(Y, \widehat{C})$ do not 
change together in different iterations of any loop enclosing $\widehat{C}$.
\end{enumerate}

So now we will derive computable criteria which will satisfy 
these two requirements 
individually and then show that together they are sufficient to 
ensure (S2). 

\subsection{Non interfering computations}

We observe that only the computations in value slice   
of $\langle V, \ell \rangle$  may  affect value of $(V, \ell)$.
So obviously we need to consider the common computations belonging to
value slice of $\langle X, \widehat{C}\rangle$ as well as to that of
$\langle Y, \widehat{C}\rangle$.
But when we want to see if values of $(X, \widehat{C})$ and $(Y, \widehat{C})$ may get affected by a common computation, we 
need to look at computations which
provide values to predicates that only control reachability of $\widehat{C}$ but
otherwise are not common to value slice of 
$\langle X, \ell \rangle$ and $\langle Y, \widehat{C}\rangle$.
This is because, a controlling predicate in a way 
restricts the values of variables participating in the predicate, along 
the $true$ and $false$ branches. If such variables, later take part in 
computations which affect the value of 
$\langle X, \widehat{C} \rangle$ 
as well as $\langle Y, \widehat{C}\rangle$ then 
the computations which provide value to such predicates, should
be considered as candidates affecting the value of 
$\langle X, \widehat{C} \rangle$ 
as well as $\langle Y, \widehat{C}\rangle$.

To illustrate, consider programs of Figure 
\ref{fig:control-interfere}.
In both programs (a) and (b), predicates
$Q_1$ at line 1 and $Q_2$ at line 4 are controlling the reachability of 
predicate computing point $\widehat{C}$ at line 7. We notice that, 
in program (a), through value of variable \texttt{z} at $Q_1$, 
predicate $Q_1$ will restrict the value of \texttt{x}
at $\widehat{C}$. Similarly, value of $z$ at $Q_1$  will also play a role 
in predicate $Q_2$ restricting the value of \texttt{y} at $\widehat{C}$ 
(assuming value of
\texttt{z} does not get reassigned from $Q_1$ to $Q_2$). As a result, in
program (a), 
computation of \texttt{z} before $Q_1$ is relating value of \texttt{x}
and \texttt{y} at $\widehat{C}$ through predicates $Q_1$ and $Q_2$. 
However, in 
program (b), it is value of variable \texttt{z1} which plays a role in
predicate $Q_1$ restricting value of \texttt{x} at $\widehat{C}$ and it
is value of variable \texttt{z2} which plays a role in predicate $Q_2$ 
restricting value of \texttt{y} at $\widehat{C}$.
As a result value of \texttt{x} and \texttt{y} at $\widehat{C}$ do not get
related due to $Q_1$ and $Q_2$ in program (b).

To find out whether some computation relates values of $(X, \widehat{C})$ 
and $(Y, \widehat{C})$, we will extend the concept of {\em value impacting} 
statement and call it {\em extended value impacting}.
A predicate which controls the reachability of point of interest
and uses same definition
of a variable which is used by a {\em value impacting node},
is also treated as a {\em value impacting node}. 
For example, in program(a) of Figure 
\ref{fig:control-interfere}, predicate $Q_2$ uses same value of $y$ as
that used in  value impacting assignment \texttt{y=y+1} for 
$\langle {y}, \widehat{C}\rangle$. Therefore, predicate $Q_2$ is also 
considered as  value impacting for $\langle {y}, \widehat{C}\rangle$. 
Similarly predicate
$Q_1$ will be considered as value impacting for 
$\langle{x}, \widehat{C}\rangle$.

\begin{definition}{\em (Extended Value-impacting node)}
\label{vic_def}
A node $s$ {\em extended-value-impacts} $\Upsilon=\langle l, V \rangle$, 
if any of the following conditions hold:\\
1.  $s$ is an assignment in $DU(\Upsilon)$.\\
2.  $s$ is an assignment, and there exists  a node $t$ such that
$t$ {\em extended-value-impacts} $\Upsilon$ and $s$ is in $DU(LV(t))$.\\
3. $s$ is a predicate $c$  from which there exist paths $\pi_1$ and
  $\pi_2$ starting with  the out-edges of $c$ and  ending at the first
  occurrence of  $l$.  Further,  there exists a  node $t \neq  c$ such
  that $t$  {\em extended-value-impacts} $\Upsilon$,  and (a) $t$ is  the 
  first
  value-impacting  node  along  $\pi_1$  (b)  $t$  is  not  the  first
  value-impacting node along $\pi_2$.\\
4. $s$ is a predicate $c$  which transitively controls
  $l$.  Further,  there exists an  {\em extended-value-impacting} node 
  $t \neq  c$ and an assignment $d$
  such that $d$ is in  $DU(LV(t))$ and $d$ is in $DU(LV(c))$.
\end{definition}

A slice made up from {\em extended value impacting} statements, will be 
called {\em extended value slice}. 
We know that in a program which is already (backward) sliced with respect 
to $\langle V, \ell \rangle$,
only variables which are live at $ENTRY$ point dictate the value 
of $(V, \ell)$ and
reachability of $\ell$. Similarly, in an extended value slice with respect 
to $\langle V, \ell \rangle$, variables live at $ENTRY$ point will dictate 
value of $(V, \ell)$, whenever $\ell$ is reached.
For extended value slice we will call these variable as value base 
of $(V, \ell)$
and denote them as $VB(V, \ell)$. We will denote the set of extended value 
impacting nodes for $\langle V, \ell\rangle$ by $EVI(V, \ell)$.

\begin{figure}
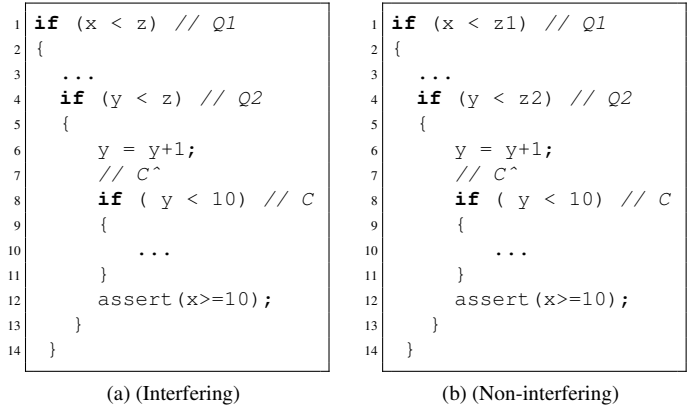

\begin{minipage}{1.50in}
\subfloat[(Interfering)]{%
\begin{lstlisting}[style=MyListStyle]
if (x < z) // Q1
{
  ...
  if (y < z) // Q2
  {
     y = y+1;
     // C^
     if ( y < 10) // C
     {
        ...
     }
     assert(x>=10);
   }
 }
\end{lstlisting}%
}
\end{minipage}
\hspace{0.30in}
\begin{minipage}{1.50in}
\subfloat[(Non-interfering)]{%
\begin{lstlisting}[style=MyListStyle]
if (x < z1) // Q1
{
  ...
  if (y < z2) // Q2
  {
     y = y+1;
     // C^
     if ( y < 10) // C
     {
        ...
     }
     assert(x>=10);
   }
 }
\end{lstlisting}%
}
\end{minipage}
\caption{Interference of control conditions with values}
\label{fig:control-interfere}
\end{figure}

We claim that if set of value base variables for 
$\langle X, \widehat{C}\rangle$ and $\langle Y, \widehat{C}\rangle$ are
disjoint then no computation affects value of 
both when $\widehat{C}$ is not enclosed in a loop.

(C3) $VB(X, \widehat{C}) \cap VB(Y, \widehat{C}) = \emptyset$

\subsection{Avoiding relation through loops}
To find a criterion which ensures that value of 
$\langle X, \widehat{C} \rangle$ and 
$\langle Y, \widehat{C} \rangle$ do not become related 
due to they changing together in a loop $L$, we need to see how to identify
if some variable is changing in a loop.

Suppose we observe value of a set of variables $Z$ at a point $\ell$ 
which is inside a loop $L$. Let $\llbracket L \rrbracket$ represent the 
body of loop $L$ and $VB(L)$ denote the value base variables of 
loop controlling condition of loop $L$.
Let $EVI(Z, \ell)$ be set of extended value impacting statements for 
$\langle Z, \ell \rangle$. 
We claim that value of $Z$ at $\ell$ may change with different
iterations of loop $L$ only if at least one  extended value impacting 
statement from 
$EVI(Z, \ell)$ is outside the body of loop $L$ and at least one such 
statement is inside the loop body. Let us call such loops
as {\em value changing loops} for $\langle Z, \ell \rangle$. That is
if $EVI(Z, \ell) \cap \llbracket L \rrbracket \neq \emptyset$ and 
$EVI(Z, \ell) \subseteq \llbracket L \rrbracket$ then 
only $\langle
Z, \ell \rangle$ can change its value in different iterations of loop $L$.

So to ensure that value of $\langle X, \widehat{C} \rangle$ and 
$\langle Y, \widehat{C} \rangle$
do not change together in any enclosing loop of $\widehat{C}$, following
criterion (C4) should be satisfied.

\begin{enumerate}
\item[(C4)]{}For all loop $L$ enclosing $\widehat{C}$, one of the 
following should hold.
\begin{enumerate}
\item[(a)]{}
$(EVI(X, \widehat{C}) \cap \llbracket L \rrbracket = \emptyset) \vee 
(EVI(X, \widehat{C}) \subseteq \llbracket L \rrbracket)$
\item[(b)]{}
$(EVI(Y, \widehat{C}) \cap \llbracket L \rrbracket = \emptyset) \vee 
(EVI(Y, \widehat{C} \rangle) \subseteq \llbracket L \rrbracket) $
\end{enumerate}
\end{enumerate}

We observe that, if value of $(V, \ell)$ changes in different
iterations of a loop $L$ then loop controlling conditions of all outer
loops (enclosing loop $L$) will be value impacting node of 
$\langle V, \ell \rangle$. Based on this observation and criterion (C4)
    about value changing loops, we have following properties for 
    value changing loops of $(X, \widehat{C})$ and $(Y, \widehat{C})$.

\begin{enumerate}
\item[(P1)]{}
if $(X,\widehat{C})$
changes in a loop $L$ then for all outer loops $L'$ of $L$, 
$VB(L') \cap VB(X, \widehat{C}) = \emptyset \vee VB(L') \subseteq VB(X, \widehat{C})$.
\item[(P2)]{}
if $(Y,\widehat{C})$
changes in a loop $L$ then for all outer loops $L'$ of $L$, 
$VB(L') \cap VB(Y, \widehat{C}) = \emptyset \vee VB(L') \subseteq VB(Y, \widehat{C})$.
\item[(P3)]{}if $(X,\widehat{C})$ changes in a loop $L_1$ and $(Y,\widehat{C})$ changes 
in a loop $L_2$ then $VB(L_1)$ and $VB(L_2)$ are disjoint.
\item[(P4)]{}
if $(X,\widehat{C})$ changes in an 
inner loop 
then $VB(L) \subseteq VB(X, \widehat{C})$.
\item[(P5)]{}
if $(Y,\widehat{C})$ changes in an 
inner loop 
then $VB(L) \subseteq VB(Y, \widehat{C})$.
\end{enumerate}



\begin{claim}
\label{cond_for_independence}
If criteria (C3) and (C4) hold for a predicate $C$ then following will hold:
$
     \forall b \in \{true, false\} . \Sigma_b \neq \emptyset \implies
     \lfloor \Sigma_b \rfloor_X = \lfloor \Sigma \rfloor_X 
     $
\end{claim}

\begin{proof}
Let $\widetilde{X}$ and $\widetilde{Y}$ be the set of value base 
variables for $\langle X, \widehat{C}\rangle$
and $\langle Y, \widehat{C} \rangle$ respectively. In addition let 
$V$ be the set of all variables.
Let $\widetilde{Z} = V - (\widetilde{X} \cup \widetilde{Y})$.
Obviously, $\widetilde{X}, \widetilde{Y}$ and $\widetilde{Z}$ are disjoint.

Assume that $b$ is $true$ and 
$\Sigma_{true} \neq \emptyset$. 
Suppose $\sigma \in \Sigma$ and $\sigma_t \in \Sigma_{true}$. 
Let input state $I_1$ produce $\sigma$ and $I_2$ produce $\sigma_t$.
We partition $I_1$ into
$\lfloor I_1\rfloor_{\widetilde{X}}, \lfloor I_1 \rfloor_{\widetilde{Y}}$ 
and $\lfloor I_1 \rfloor_{\widetilde{Z}}$. 
Similarly, we partition $I_2$
into 
$\lfloor I_2\rfloor_{\widetilde{X}}, \lfloor I_2 \rfloor_{\widetilde{Y}}$ 
and $\lfloor I_2 \rfloor_{\widetilde{Z}}$. 

Let $VB(L)$ denote the base variables for loop controlling conditions of 
loop $L$.  
We construct an input $I_3$ such that values of $\widetilde{X}$ come from
$I_1$ and those for $\widetilde{Y}$ come from $I_2$.
For the remaining values we proceed as follows:

For each loop $L$ enclosing $\widehat{C}$, if $VB(L)$ is not included in
$\widetilde{X}$ or $\widetilde{Y}$ then we proceed as follows:
\begin{enumerate}
\item If $(X,\widehat{C})$ changes in $L$ then use values of $VB(L)$ from
$I_1$.
\item If $(Y,\widehat{C})$ changes in $L$ then use values of $VB(L)$ from
$I_2$.
\item If none of $(X,\widehat{C})$ and $(Y,\widehat{C})$ changes in $L$
then use values of $VB(L)$ from $I_1$
\end{enumerate}
By the properties (P1) to (P5), such assignments will be possible without 
any 
conflict. That is no variable will be required to have it value from $I_1$
as well as from $I_2$.
By the above step, some of variables from $\widetilde{Z}$ would have 
got their input values. For the remaining variables of $\widetilde{Z}$, if
any, take their values from $I_1$.
Obviously starting with input $I_3$, 
$\widehat{\mathcal{C}}$ will be reachable,
with value of $X$ at $\widehat{\mathcal{C}}$ same as  that with $I_1$ and 
value of $Y$ at $\widehat{\mathcal{C}}$ same as  
that with $I_2$. Therefore, one of the states produced with input $I_3$ at
$\widehat{C}$ will have same values of $X$ as that in $\sigma$ and same 
values of $Y$ as that in $\sigma_t$.
Let $\sigma_3$ be such a state. Obviously $C$ will evaluate to $true$ in 
state $\sigma_3$ and therefore $\sigma_3 \in \Sigma_{true}$.
Moreover $\lfloor \sigma \rfloor_X = \lfloor \sigma_3 \rfloor_X$.
Therefore $\lfloor \Sigma \rfloor_X \subseteq \lfloor \Sigma_{true} \rfloor_X$.
similarly we can show that $\lfloor \Sigma \rfloor_X \subseteq \lfloor \Sigma_{false} \rfloor_X$.
\end{proof}

It is obvious that checking the criteria (C1), (C2), (C3) and (C4) is 
computable for a given predicate $C$.

\section{Property checking with \ITPonly\ predicates}
Assume that we identified an \ITPonly\ predicate $C$ in program $P$ as
per the definition ~\ref{itp_def}.
We abstract $P$ to $P'$ using the abstraction strategy mentioned earlier. 
Now we run a property checker on program $P'$. We consider following cases:

\begin{enumerate}
\item[(i)] Property holds in $P'$.
\item[(ii)] Property gets violated with a counter example 
bypassing $\widehat{C}$. 
\item[(iii)] Property is violated in $P'$ with counter example passing 
through $\widehat{C}$.
\end{enumerate}
For case (i), we are done as the property will hold in program $P$ also. 
For case (ii),  the property will get violated in $P$ also with
same counter example as that of $P'$, as per the property of the 
abstraction mechanism. 
The case (iii) needs to be analysed further and we
proceed as follows:

Let input $I'$ be the the counter example, producing trace $\tau'$,
for abstract program $P'$. Since the counter example trace passes
through $\widehat{C}$, it will have a violatinng-suffix, say $\pi$.
We execute program $P$ with same input $I'$ to get trace $\tau$ and  
consider following possible outcomes for trace $\tau$. 

\begin{enumerate}
\item Assertion is violated.
\item Assertion is not violated and trace $\tau'$ bypasses 
the condition $C$.
\item Assertion is not violated in trace $\tau$ and trace $\tau'$ passes 
through $C$ with violating-value of $C$ 
as $b \in \lbrace true, false \rbrace$.
\end{enumerate}
If it is case (1) then we are done as we found a counter example in 
concrete program $P$ also.
For case (3) we consider following sub cases:
\begin{enumerate}
\item[3(a)] Trace $\tau$ never passed through $C$.
\item[3(b)] Trace $\tau$, passed through
$C$ and  $C$ evaluated to $b$ at least once.
\item[3(c)] Trace $\tau$ passed through
$C$ but $C$ always evaluated to $\neg b$. 
\end{enumerate}
If it is case 2, 3(a) or 3(c), we want to check the possibility of $C$ 
always evaluating to $\neg b$ in program $P$. For this, we create a program 
$\widehat{P}$ from $P$
by placing a new assert expression $(C==\neg b)$ at $\widehat{C}$ and 
removing 
the old assert. We also replace predicate $C$ with $\neg b$. We claim
that assertion in $\widehat{P}$ will hold if and only if $C$ always 
evaluated to $\neg b$ in program $P$.
We solve the new property checking problem $\widehat{P}$
and consider following possible outcomes.
\begin{enumerate}
\item[(A)] New property does not hold in $\widehat{P}$, implying that the 
predicate $C$ evaluates to $b$ also sometimes.
\item[(B)] New property in $\widehat{P}$ holds implying $C$ always evaluates to 
$\neg b$.
\end{enumerate}
In case (B), we create a new property checking problem $\widetilde{P}$ 
from 
$P$ by replacing $C$ with $\neg b$. Problem $\widetilde{P}$ will be 
property equivalent to 
$P$. Solving $\widetilde{P}$ will give solution to $P$.
 
For the cases (A) and 3(b), by definition ~\ref{itp_def} of \ITPonly\ 
predicates, we know that, for program $P$,  there exists a counter 
example trace
which has a violating-suffix matching with violating-suffix $\pi$ of
counter example trace of $P'$.
We compute a weakest-precondition $\psi$ of $\neg A$ 
for the path $\pi$.
Obviously, $\psi$ must be satisfiable at $\widehat{C}$ in $P$.
So there must exist an input for program $P$ which satisfies $\psi$ at 
$\widehat{C}$ and the same will indeed
be a countre example in $P$ for assertion $A$.

So we just need to find a counter
example which violates the assertion $\neg \psi$ at $\widehat{C}$.
So, we create an input generation problem as property checking for the 
assertion $\neg \psi$ in program $P$ at $\widehat{C}$. 
Obviously, this assertion
must get violated and if verifier finds an 
input for the same, it will
be counter example for the original property checking problem.

\balance


\end{document}